\documentclass[11pt]{article}
\usepackage{graphicx}
\usepackage{tabularx}
\usepackage{url}
\usepackage{amsthm}
\usepackage{amsmath}
\usepackage{amsfonts}
\usepackage{amssymb}

\theoremstyle{plain}
\newtheorem{theorem}{Theorem}
\newtheorem{lemma}{Lemma}

\theoremstyle{definition}

\theoremstyle{remark}

\date{}

\begin{document}

\title{Exact learning and test theory}
\author{Mikhail Moshkov\thanks{Computer, Electrical and Mathematical Sciences and Engineering Division,
King Abdullah University of Science and Technology (KAUST),
Thuwal 23955-6900, Saudi Arabia. Email: mikhail.moshkov@kaust.edu.sa.
}}
\maketitle

\begin{abstract}
In this paper, based on results of exact learning and test theory, we study
arbitrary infinite binary information systems each of which consists of an
infinite set of elements and an infinite set of two-valued functions
(attributes) defined on the set of elements. We consider the notion of a
problem over information system, which is described by a finite number of
attributes: for a given element, we should recognize values of these
attributes. As algorithms for problem solving, we consider decision trees of
two types: (i) using only proper hypotheses (an analog of proper equivalence
queries from exact learning), and (ii) using both attributes and proper
hypotheses. As time complexity, we study the depth of decision trees. In the
worst case, with the growth of the number of attributes in the problem
description, the minimum depth of decision trees of both types either is
bounded from above by a constant or grows as a logarithm, or linearly. Based
on these results and results obtained earlier for attributes and arbitrary
hypotheses, we divide the set of all infinite binary information systems
into seven complexity classes.
\end{abstract}


\section{Introduction \label{S1a}}

Exact learning initiated by Angluin \cite{Angluin88} and test theory
initiated by Chegis and Yablonskii \cite{Chegis58} both study decision
trees. These theories are closely related. In particular, attributes from
test theory correspond to membership queries from exact learning. Exact
learning considers additionally so-called equivalence queries.

Heged\"{u}s in \cite{Hegedus95} generalized some bounds from \cite{Moshkov83}
obtained in the framework of test theory to the case of exact learning with
membership and equivalence queries. Similar results were obtained
independently and in the other way by Hellerstein et al. \cite{Hellerstein96}%
. In this paper, we move in the opposite direction: we add to the model
considered in test theory the notion of a hypothesis that allows us to use
an analog of equivalence queries. In \cite{DAM}, we studied arbitrary
hypotheses. This paper is devoted to the consideration of proper hypotheses
(an analog of proper equivalence queries).

We study infinite binary information systems each of which consists of an
infinite set of elements $A$ and an infinite set $F$ of functions
(attributes) from $A$ to $\{0,1\}$. We define the notion of a problem
described by a finite number of attributes $f_{1},\ldots ,f_{n}$ from $F$:
for a given element $a\in A$, we should recognize the tuple $%
(f_{1}(a),\ldots ,f_{n}(a))$. To this end, we can use decision trees based
on two types of queries. We can ask about the value of an attribute $%
f_{i}\in \{f_{1},\ldots ,f_{n}\}$. We will obtain an answer of the kind $%
f_{i}(x)=\delta $, where $\delta \in \{0,1\}$. We can also ask if a
hypothesis $f_{1}(x)=\delta _{1},\ldots ,f_{n}(x)=\delta _{n}$ is true,
where $\delta _{1},\ldots ,\delta _{n}\in \{0,1\}$. Either this hypothesis
will be confirmed or we will obtain a counterexample in the form $%
f_{i}(x)=\lnot \delta _{i}$. The considered hypothesis is called proper if
there exists an element $a\in A$ such that $f_{1}(a)=\delta _{1},\ldots
,f_{n}(a)=\delta _{n}$. As time complexity of a decision tree, we consider
its depth, which is equal to the maximum number of queries in a path from
the root to a terminal node of the tree.

Based on the results of exact learning \cite%
{Angluin88,Angluin04,Littlestone88,Maass92}, and test theory and rough set theory \cite%
{Moshkov83,Moshkov89,Moshkov05}, for an arbitrary infinite binary
information system, we studied in \cite{DAM} three functions of Shannon
type, which characterize the dependence in the worst case of the minimum
depth of a decision tree solving a problem on the number of attributes in
the problem description. The considered three functions correspond to the
following three cases:

\begin{itemize}
\item Only attributes are used in decision trees.

\item Only hypotheses are used in decision trees.

\item Both attributes and hypotheses are used in decision trees.
\end{itemize}

We proved that the first function has two possible types of behavior:
logarithmic and linear. The second and the third functions have three
possible types of behavior: constant, logarithmic, and linear. The first
function was studied in \cite{Moshkov89,Moshkov05}. Results related to the
second and the third functions were presented in \cite{Moshkov01} without
proofs. We provided these proofs in \cite{DAM}. In the same paper, we also
studied joint behavior of these three functions and described four complexity
classes of infinite binary information systems.

In this paper, we study two functions of Shannon type, which also
characterize the dependence in the worst case of the minimum depth of a
decision tree solving a problem on the number of attributes in the problem
description. These functions correspond to the following two cases:

\begin{itemize}
\item Only proper hypotheses are used in decision trees.

\item Both attributes and proper hypotheses are used in decision trees.
\end{itemize}

We prove that both functions have three possible types of behavior:
constant, logarithmic, and linear. Results related to these functions were
presented in \cite{Moshkov01} without proofs. We also study joint behavior of
all five functions and describe seven complexity classes of infinite binary
information systems.

The rest of the paper is organized as follows. In Sections \ref{S2a} and \ref%
{S3a}, we consider basic notions and previous results obtained in \cite{DAM}%
. In Section \ref{S4a}, we present main results of this paper. Sections \ref%
{S5a} and \ref{S6a} contains proofs of main results, and Section \ref{S7a}
-- short conclusions.

\section{Basic Notions \label{S2a}}

Let $A$ be a set and $F$ be a set of functions from $A$ to $\{0,1\}$.
Functions from $F$ are called \emph{attributes} and the pair $U=(A,F)$ is
called a\emph{\ binary information system} (this notion is close to the
notion of information system proposed by Pawlak \cite{Pawlak81}). If $A$ and
$F$ are infinite sets, then the pair $U=(A,F)$ is called an \emph{infinite
binary information system}.

The set $A$ may be interpreted as the set of inputs for problems over the
information system $U$. A \emph{problem over} $U$ is an arbitrary $n$-tuple $%
z=(f_{1},\ldots ,f_{n})$ where $n\in \mathbb{N} $, $\mathbb{N} $ is the set
of natural numbers $\{1,2,\ldots \}$, and $f_{1},\ldots ,f_{n}\in F$. The
problem $z$ may be interpreted as a problem of searching for the tuple $%
z(a)=(f_{1}(a),\ldots ,f_{n}(a))$ for an arbitrary $a\in A$. The number $%
\dim z=n$ is called the \emph{dimension} of the problem $z$. Denote $%
F(z)=\{f_{1},\ldots ,f_{n}\}$. We denote by $P(U)$ the set of problems over $%
U$.

A \emph{system of equations over }$U$ is an arbitrary equation system of the
kind
\[
\{g_{1}(x)=\delta _{1},\ldots ,g_{m}(x)=\delta _{m}\}
\]%
where $m\in \mathbb{N}\cup \{0\}$, $g_{1},\ldots ,g_{m}\in F$, and $\delta
_{1},\ldots ,\delta _{m}\in \{0,1\}$ (if $m=0$, then the considered equation
system is empty). This equation system is called a \emph{system of equations
over }$z$ if $g_{1},\ldots ,g_{m}\in F(z)$. The considered equation system
is called \emph{consistent} (on $A$) if its set of solutions on $A$ is
nonempty. The set of solutions of the empty equation system coincides with $%
A $.

As algorithms for problem $z$ solving, we consider decision trees with two
types of \emph{queries}. We can choose an attribute $f_{i}\in F(z)$ and ask
about its value. This query has two possible answers $\{f_{i}(x)=0\}$ and $%
\{f_{i}(x)=1\}$. We can formulate a \emph{hypothesis over} $z$ in the form $%
H=\{f_{1}(x)=\delta _{1},\ldots ,f_{n}(x)=\delta _{n}\}$ where $\delta
_{1},\ldots ,\delta _{n}\in \{0,1\}$, and ask about this hypothesis. This
query has $n+1$ possible answers: $H,\{f_{1}(x)=\lnot \delta
_{1}\},...,\{f_{n}(x)=\lnot \delta _{n}\}$ where $\lnot 1=0$ and $\lnot 0=1$%
. The first answer means that the hypothesis is true. Other answers are
counterexamples. The hypothesis $H$ is called \emph{proper }(for $U$) if the
system of equations $H$ is consistent on $A$.

A \emph{decision tree over} $z$ is a marked finite directed tree with the
root in which

\begin{itemize}
\item Each \emph{terminal} node is labeled with an $n$-tuple from the set $%
\{0,1\}^{n}$.

\item Each node, which is not terminal (such nodes are called \emph{working}%
), is labeled with an attribute from the set $F(z)$ or with a hypothesis
over $z$.

\item If a working node is labeled with an attribute $f_{i}$ from $F(z)$,
then there are two edges, which leave this node and are labeled with the
systems of equations $\{f_{i}(x)=0\}$ and $\{f_{i}(x)=1\}$, respectively.

\item If a working node is labeled with a hypothesis
\[
H=\{f_{1}(x)=\delta _{1},\ldots ,f_{n}(x)=\delta _{n}\}
\]
over $z$, then there are $n+1$ edges, which leave this node and are labeled
with the systems of equations $H, \{f_{1}(x)=\lnot \delta _{1}\}, ... ,
\{f_{n}(x)=\lnot \delta _{n}\}$, respectively.
\end{itemize}

Let $\Gamma $ be a decision tree over $z$. A \emph{complete path} in $\Gamma
$ is an arbitrary directed path from the root to a terminal node in $\Gamma $%
. We now define an equation system $\mathcal{S}(\xi )$ over $U$ associated
with the complete path $\xi $. If there are no working nodes in $\xi $, then
$\mathcal{S}(\xi )\ $is the empty system. Otherwise, $\mathcal{S}(\xi )$ is
the union of equation systems assigned to the edges of the path $\xi $. We
denote by $\mathcal{A}(\xi )=\mathcal{A}_A(\xi )$ the set of solutions on $A$ of the system of
equations $\mathcal{S}(\xi )$ (if this system is empty, then its solution
set is equal to $A$).

We will say that a decision tree $\Gamma $ over $z$ \emph{solves\ the
problem }$z$\emph{\ relative to} $U$ if, for each element $a\in A$ and for
each complete path $\xi $ in $\Gamma $ such that $a\in \mathcal{A}(\xi )$,
the terminal node of the path $\xi $ is labeled with the tuple $z(a)$.

We now consider an equivalent definition of a decision tree solving a
problem. Denote by $\Delta _{U}(z)$ the set of tuples $(\delta _{1},\ldots
,\delta _{n})\in \{0,1\}^{n}$ such that the system of equations $%
\{f_{1}(x)=\delta _{1},\ldots ,f_{n}(x)=\delta _{n}\}$ is consistent. The
set $\Delta _{U}(z)$ is the set of all possible solutions to the problem $z$%
. Let $\Delta \subseteq \Delta _{U}(z)$, $f_{i_{1}},\ldots ,f_{i_{m}}\in
\{f_{1},\ldots ,f_{n}\}$, and $\sigma _{1},\ldots ,\sigma _{m}\in \{0,1\}$.
Denote
\[
\Delta (f_{i_{1}},\sigma _{1})\cdots (f_{i_{m}},\sigma _{m})
\]
the set of all $n$-tuples $(\delta _{1},\ldots ,\delta _{n})\in \Delta $ for
which $\delta _{i_{1}}=\sigma _{1},\ldots ,\delta _{i_{m}}=\sigma _{m}$.

Let $\Gamma $ be a decision tree over the problem $z$. We correspond to each
complete path $\xi $ in the tree $\Gamma $ a word $\pi (\xi )$ in the
alphabet $\{(f_{i},\delta ):f_{i}\in F(z),\delta \in \{0,1\}\}$. If the
equation system $\mathcal{S}(\xi )$ is empty, then $\pi (\xi )$ is the empty
word. If $\mathcal{S}(\xi )=\{f_{i_{1}}(x)=\sigma _{1},\ldots
,f_{i_{m}}(x)=\sigma _{m}\}$, then $\pi (\xi )=(f_{i_{1}},\sigma _{1})\cdots
(f_{i_{m}},\sigma _{m})$. The decision tree $\Gamma $ over $z$ solves the
problem $z$ relative to $U$ if, for each complete path $\xi $ in $\Gamma $,
the set $\Delta _{U}(z)\pi (\xi )$ contains at most one tuple and if this
set contains exactly one tuple, then the considered tuple is assigned to the
terminal node of the path $\xi $.

As time complexity of a decision tree, we consider its \emph{depth} that is
the maximum number of working nodes in a complete path in the tree or, which
is the same, the maximum length of a complete path in the tree. We denote by
$h(\Gamma )$ the depth of a decision tree $\Gamma $.

Let $z\in P(U)$. We denote by $h_{U}^{(1)}(z)$ the minimum depth of a
decision tree over $z$, which solves $z$ relative to $U$ and uses only
attributes from $F(z)$. We denote by $h_{U}^{(2)}(z)$ the minimum depth of a
decision tree over $z$, which solves $z$ relative to $U$ and uses only
hypotheses over $z$. We denote by $h_{U}^{(3)}(z)$ the minimum depth of a
decision tree over $z$, which solves $z$ relative to $U$ and uses both
attributes from $F(z)$ and hypotheses over $z$. We denote by $h_{U}^{(4)}(z)$
the minimum depth of a decision tree over $z$, which solves $z$ relative to $%
U$ and uses only proper hypotheses over $z$. We denote by $h_{U}^{(5)}(z)$
the minimum depth of a decision tree over $z$, which solves $z$ relative to $%
U$ and uses both attributes from $F(z)$ and proper hypotheses over $z$.

For $i=1,2,3,4,5$, we define a function of Shannon type $h_{U}^{(i)}(n)$
that characterizes dependence of $h_{U}^{(i)}(z)$ on $\dim z$ in the worst
case. Let $i\in \{1,2,3,4,5\}$ and $n\in \mathbb{N}$. Then%
\[
h_{U}^{(i)}(n)=\max \{h_{U}^{(i)}(z):z\in P(U),\dim z\leq n\}.
\]

\section{Previous Results \label{S3a}}

In this section, we consider results presented in \cite{DAM}.

Let $U=(A,F)$ be an infinite binary information system and $r\in \mathbb{N}$%
. We will say that the information system $U$ is $r$-\emph{reduced} if, for
each consistent on $A$ system of equations over $U$, there exists a
subsystem of this system that has the same set of solutions and contains at
most $r$ equations. We denote by $\mathcal{R}$ the set of infinite binary
information systems each of which is $r$-reduced for some $r\in
\mathbb{N}$.

The next theorem follows from results obtained in \cite{Moshkov89} where we
considered closed classes of test tables (decision tables). It also follows
from the results obtained in \cite{Moshkov05} where we considered the
weighted depth of decision trees.  In particular, the upper bound
mentioned in item (a) of the theorem is based on a halving algorithm that
is similar to proposed in \cite{Moshkov83}.

\begin{theorem}
\label{T1a} {\rm \cite{DAM}} Let $U$ be an infinite binary information system.
Then the following statements hold:

(a) If $U\in \mathcal{R}$, then $h_{U}^{(1)}(n)=\Theta (\log n)$.

(b) If $U\notin \mathcal{R}$, then $h_{U}^{(1)}(n)=n$ for any $n\in \mathbb{N%
}$.
\end{theorem}

A subset $\{f_{1},\ldots ,f_{m}\}$ of $F$ is called \emph{independent} if,
for any $\delta _{1},\ldots ,\delta _{m}\in \{0,1\}$, the system of
equations $\{f_{1}(x)=\delta _{1},\ldots ,f_{m}(x)=\delta _{m}\}$ is
consistent on the set $A$. The empty set of attributes is independent by
definition. We now define the parameter $I(U)$, which is called the \emph{%
independence dimension} or $I$-\emph{dimension} of the information system $U$
(this notion is similar to the notion of independence number of family of
sets considered by Naiman and Wynn in \cite{Naiman96}). If, for each $m\in
\mathbb{N}$, the set $F$ contains an independent subset of the cardinality $%
m $, then $I(U)=\infty $. Otherwise, $I(U)$ is the maximum cardinality of an
independent subset of the set $F$. We denote by $\mathcal{D}$ the set of
infinite binary information systems with finite independence dimension.

Let $U=(A,F)$ be a binary information system, which is not necessary
infinite, $f\in F$, and $\delta \in \{0,1\}$. Denote
\[
A(f,\delta )=\{a:a\in A,f(a)=\delta \}.
\]%
We now define inductively the notion of $k$-\emph{information system}, $k\in
\mathbb{N}\cup \{0\}$. The binary information system $U$ is called $0$%
-information system if all attributes from $F$ are constant on the set $A$.
Let, for some $k\in \mathbb{N}\cup \{0\}$, the notion of $m$-information
system be defined for $m=0,\ldots ,k$. The binary information system $U$ is
called $(k+1)$-information system if it is not $m$-information system for $%
m=0,\ldots ,k$ and, for any $f\in F$, there exist numbers $\delta \in
\{0,1\} $ and $m\in \{0,\ldots ,k\}$ such that the information system $%
(A(f,\delta ),F)$ is $m$-information system. It is easy to show by induction
on $k$ that if $U=(A,F)$ is $k$-information system, then $U^{\prime
}=(A^{\prime },F)$, $A^{\prime }\subseteq A$, is $l$-information system for
some $l\leq k$. We denote by $\mathcal{C}$ the set of infinite binary
information systems for each of which there exists $k\in \mathbb{N}$ such
that the considered system is $k$-information system.

We proved in \cite{DAM} that $\mathcal{C\subseteq D}$. Therefore, for any
infinite binary information system $U$, either $U\in \mathcal{C}$, or $%
U\in \mathcal{D}\setminus \mathcal{C}$, or $U\notin \mathcal{D}$.

The following theorem
was presented in \cite{Moshkov01} without proof. We gave the proof of this
theorem in \cite{DAM}. Note that the lower bounds mentioned in item (b) of
the theorem were obtained by methods similar to used by Littlestone \cite%
{Littlestone88}, Maass and Tur{\'{a}}n \cite{Maass92}, and Angluin \cite%
{Angluin04} (see, for example, Lemma \ref{L0b}).

\begin{theorem}
\label{T2a} {\rm \cite{DAM}} Let $U$ be an infinite binary information system.
Then the following statements hold:

(a) If $U\in \mathcal{C}$, then $h_{U}^{(2)}(n)=O(1)$ and $%
h_{U}^{(3)}(n)=O(1)$.

(b) If $U\in \mathcal{D}\setminus \mathcal{C}$, then $h_{U}^{(2)}(n)=\Theta
(\log n)$, $h_{U}^{(3)}(n)=\Omega (\frac{\log n}{\log \log n})$, and $%
h_{U}^{(3)}(n)=O(\log n)$.

(c) If $U\notin \mathcal{D}$, then $h_{U}^{(2)}(n)=n$ and $h_{U}^{(3)}(n)=n$
for any $n\in \mathbb{N}$.
\end{theorem}

Let $U$ be an infinite binary information system. We now consider the joint
behavior of the functions $h_{U}^{(1)}(n)$, $h_{U}^{(2)}(n)$, and $%
h_{U}^{(3)}(n)$. It depends on the belonging of the information system $U$
to the sets $\mathcal{R}$, $\mathcal{D}$, and $\mathcal{C}$. We correspond
to the information system $U$ its \emph{indicator vector} $%
ind(U)=(c_{1},c_{2},c_{3})\in \{0,1\}^{3}$ in which $c_{1}=1$ if and only if
$U\in \mathcal{R}$, $c_{2}=1$ if and only if $U\in \mathcal{D}$, and $%
c_{3}=1 $ if and only if $U\in \mathcal{C}$. Proof of the following theorem
is presented in \cite{DAM}.

\begin{table}[h]
\caption{Possible indicator vectors of infinite binary information systems}
\label{tab1}\center
\begin{tabular}{|l|lll|}
\hline
& $\mathcal{R}$ & $\mathcal{D}$ & $\mathcal{C}$ \\ \hline
1 & $0$ & $0$ & $0$ \\
2 & $0$ & $1$ & $0$ \\
3 & $0$ & $1$ & $1$ \\
4 & $1$ & $1$ & $0$ \\ \hline
\end{tabular}%
\end{table}

\begin{theorem}
\label{T3a} {\rm \cite{DAM}} For any infinite binary information system, its
indicator vector coincides with one of the rows of Table \ref{tab1}. Each
row of Table \ref{tab1} is the indicator vector of some infinite binary
information system.
\end{theorem}

\begin{table}[h]
\caption{Summary of Theorems \protect\ref{T1a}-\protect\ref{T3a}}
\label{tab2}\center
\begin{tabular}{|l|lll|lll|}
\hline
& $\mathcal{R}$ & $\mathcal{D}$ & $\mathcal{C}$ & $h_{U}^{(1)}(n)$ & $%
h_{U}^{(2)}(n)$ & $h_{U}^{(3)}(n)$ \\ \hline
$V_{1}$ & $0$ & $0$ & $0$ & $n$ & $n$ & $n$ \\
$V_{2}$ & $0$ & $1$ & $0$ & $n$ & $\Theta (\log n)$ & $\Omega (\frac{\log n}{%
\log \log n}),O(\log n)$ \\
$V_{3}$ & $0$ & $1$ & $1$ & $n$ & $O(1)$ & $O(1)$ \\
$V_{4}$ & $1$ & $1$ & $0$ & $\Theta (\log n)$ & $\Theta (\log n)$ & $\Omega (%
\frac{\log n}{\log \log n}),O(\log n)$ \\ \hline
\end{tabular}%
\end{table}

For $i=1,2,3,4$, we denote by $V_{i}$ the class of all infinite binary
information systems, which indicator vector coincides with the $i$th row of
Table \ref{tab1}. Table \ref{tab2} summarizes Theorems \ref{T1a}-\ref{T3a}.
The first column contains the name of \emph{complexity class} $V_{i}$. The
next three columns describe the indicator vector of information systems from
this class. The last three columns $h_{U}^{(1)}(n)$, $h_{U}^{(2)}(n)$, and $%
h_{U}^{(3)}(n)$ contain information about behavior of the functions $%
h_{U}^{(1)}(n)$, $h_{U}^{(2)}(n)$, and $h_{U}^{(3)}(n)$ for information
systems from the class $V_{i}$.

\section{Main Results \label{S4a}}

In this section, we consider main results of this paper.

Let $U=(A,F)$ be a binary information system and $r\in \mathbb{N}$. We will
say that the information system $U$ is $r$\emph{-i-reduced} if, for each
inconsistent on $A$ system of equations over $U$, there exists a subsystem
of this system that is inconsistent and contains at most $r$ equations. We
denote by $\mathcal{I}$ the set of infinite binary information systems each
of which is $r$-i-reduced for some $r\in \mathbb{N}$.

We proved in \cite{DAM} that $\mathcal{C\subseteq D}$. Therefore, for any
infinite binary information system $U$, either $U\in \mathcal{C\cap I}$, or $%
U\in (\mathcal{D}\setminus \mathcal{C)\cap I}$, or $U\in \mathcal{D}%
\setminus \mathcal{I}$, or $U\notin \mathcal{D}$.
The following theorem was presented in \cite{Moshkov01} without proof.

\begin{theorem}
\label{T4a}Let $U$ be an infinite binary information system. Then the
following statements hold:

(a) If $U\in \mathcal{C\cap I}$, then $h_{U}^{(4)}(n)=O(1)$ and $%
h_{U}^{(5)}(n)=O(1)$.

(b) If $U\in (\mathcal{D}\setminus \mathcal{C)\cap I}$, then $%
h_{U}^{(4)}(n)=\Theta (\log n)$, $h_{U}^{(5)}(n)=\Omega (\frac{\log n}{\log
\log n})$, and $h_{U}^{(5)}(n)=O(\log n)$.

(c) If $U\in \mathcal{D}\setminus \mathcal{I}$ and $i\in \{4,5\}$, then $%
h_{U}^{(i)}(n)\geq n-1$ for infinitely many $n\in \mathbb{N}$ and $%
h_{U}^{(i)}(n)\leq n$ for any $n\in \mathbb{N}$

(d) If $U\notin \mathcal{D}$, then $h_{U}^{(4)}(n)=n$ and $h_{U}^{(5)}(n)=n$
for any $n\in \mathbb{N}$.
\end{theorem}

Let $U$ be an infinite binary information system. We now consider the joint
behavior of the functions $h_{U}^{(1)}(n)$, $h_{U}^{(2)}(n)$, $%
h_{U}^{(3)}(n) $, $h_{U}^{(4)}(n)$, and $h_{U}^{(5)}(n)$. It depends on the
belonging of the information system $U$ to the sets $\mathcal{R}$, $\mathcal{%
D}$, $\mathcal{C}$, and $\mathcal{I}$. We correspond to the information
system $U$ its \emph{extended indicator vector} $%
eind(U)=(c_{1},c_{2},c_{3},c_{4})\in \{0,1\}^{4}$ in which $c_{1}=1$ if and
only if $U\in \mathcal{R}$, $c_{2}=1$ if and only if $U\in \mathcal{D}$, $%
c_{3}=1$ if and only if $U\in \mathcal{C} $, and $c_{4}=1$ if and only if $%
U\in \mathcal{I}$.

\begin{table}[h]
\caption{Possible extended indicator vectors of infinite binary information systems}
\label{tab3}\center
\begin{tabular}{|l|llll|}
\hline
& $\mathcal{R}$ & $\mathcal{D}$ & $\mathcal{C}$ & $\mathcal{I}$ \\ \hline
1 & $0$ & $0$ & $0$ & $0$ \\
2 & $0$ & $0$ & $0$ & $1$ \\
3 & $0$ & $1$ & $0$ & $0$ \\
4 & $0$ & $1$ & $0$ & $1$ \\
5 & $0$ & $1$ & $1$ & $0$ \\
6 & $0$ & $1$ & $1$ & $1$ \\
7 & $1$ & $1$ & $0$ & $1$ \\ \hline
\end{tabular}%
\end{table}

\begin{theorem}
\label{T5a} For any infinite binary information system, its extended
indicator vector coincides with one of the rows of Table \ref{tab3}. Each
row of Table \ref{tab3} is the extended indicator vector of some infinite
binary information system.
\end{theorem}

\begin{table}[h]
\caption{Summary of Theorems \protect\ref{T1a}, \protect\ref{T2a}, \protect\ref{T4a}, and \protect\ref{T5a}}
\label{tab4}\center
\begin{tabular}{|l|llll|lllll|}
\hline
& $\mathcal{R}$ & $\mathcal{D}$ & $\mathcal{C}$ & $\mathcal{I}$ & $%
h_{U}^{(1)}(n)$ & $h_{U}^{(2)}(n)$ & $h_{U}^{(3)}(n)$ & $h_{U}^{(4)}(n)$ & $%
h_{U}^{(5)}(n)$ \\ \hline
$\mathcal{V}_{1}$ & $0$ & $0$ & $0$ & $0$ & $n$ & $n$ & $n$ & $n$ & $n$ \\
$\mathcal{V}_{2}$ & $0$ & $0$ & $0$ & $1$ & $n$ & $n$ & $n$ & $n$ & $n$ \\
$\mathcal{V}_{3}$ & $0$ & $1$ & $0$ & $0$ & $n$ & $\Theta (\log n)$ & $%
\approx \log n$ & $\approx n$ & $\approx n$ \\
$\mathcal{V}_{4}$ & $0$ & $1$ & $0$ & $1$ & $n$ & $\Theta (\log n)$ & $%
\approx \log n$ & $\Theta (\log n)$ & $\approx \log n$ \\
$\mathcal{V}_{5}$ & $0$ & $1$ & $1$ & $0$ & $n$ & $O(1)$ & $O(1)$ & $\approx
n$ & $\approx n$ \\
$\mathcal{V}_{6}$ & $0$ & $1$ & $1$ & $1$ & $n$ & $O(1)$ & $O(1)$ & $O(1)$ &
$O(1)$ \\
$\mathcal{V}_{7}$ & $1$ & $1$ & $0$ & $1$ & $\Theta (\log n)$ & $\Theta
(\log n)$ & $\approx \log n$ & $\Theta (\log n)$ & $\approx \log n$ \\ \hline
\end{tabular}%
\end{table}

For $i=1,\ldots ,7$, we denote by $\mathcal{V}_{i}$ the class of all
infinite binary information systems, which extended indicator vector
coincides with the $i$th row of Table \ref{tab3}. Table \ref{tab4}
summarizes Theorems \ref{T1a}, \ref{T2a}, \ref{T4a}, and \ref{T5a}. The first
column contains the name of \emph{complexity class} $\mathcal{V}_{i}$. The
next four columns describe the extended indicator vector of information
systems from this class. The last five columns $h_{U}^{(1)}(n)$, ..., $%
h_{U}^{(5)}(n)$ contain information about behavior of the functions $%
h_{U}^{(1)}(n)$, ..., $h_{U}^{(5)}(n)$ for information systems from the
class $\mathcal{V}_{i}$. The notation $\approx \log n$ in a column $%
h_{U}^{(i)}(n)$ means that $h_{U}^{(i)}(n)$ $=\Omega (\frac{\log n}{\log
\log n})$ and $h_{U}^{(i)}(n)$ $=O(\log n)$. The notation $\approx n$ in a
column $h_{U}^{(i)}(n)$ means that $h_{U}^{(i)}(n)$ $\leq n$ for any $n\in
\mathbb{N}
$ and $h_{U}^{(i)}(n)$ $\geq n-1$ for infinitely many $n\in
\mathbb{N}
$.

Note that it is possible to consider the union $\mathcal{V}_{1}\cup
\mathcal{V}_{2}$ of the complexity classes $\mathcal{V}_{1}$ and $\mathcal{V}%
_{2}$ as one complexity class since the functions $h_{U}^{(1)}(n),\ldots
,h_{U}^{(5)}(n)$ have the same behavior for information systems from these
classes. However, in this paper, we study $\mathcal{V}_{1}$ and $\mathcal{V%
}_{2}$ as different complexity classes.

\section{Proof of Theorem \protect\ref{T4a} \label{S5a}}

In this section, we prove Theorem \ref{T4a}. First, we consider several
auxiliary statements. The following result was obtained in \cite{DAM}.

\begin{lemma}
\label{L1a} {\rm\cite{DAM}} $\mathcal{C}\subseteq \mathcal{D}$.
\end{lemma}

From this lemma it follows that, for any infinite binary information system $U$,
either $U\in \mathcal{C\cap I}$, or $U\in (\mathcal{D}\setminus \mathcal{%
C)\cap I}$, or $U\in \mathcal{D}\setminus \mathcal{I}$, or $U\notin \mathcal{%
D}$.

\begin{lemma}
\label{L2a} Let $U=(A,F)$ be a binary information system, $z$ be a problem
over $U$, and $\Gamma _{1}$ be a decision tree over $z$ that solves\ the
problem $z$\ relative to $U$ and uses both attributes from $F(z)$ and proper
hypotheses over $z$. Then there exists a decision tree $\Gamma _{2}$ over $z$
that solves\ the problem $z$\ relative to $U$, uses only proper hypotheses
over $z$, and satisfies the inequality $h(\Gamma _{2})\leq 2^{h(\Gamma
_{1})}-1$.
\end{lemma}

\begin{proof}
We prove this statement by the induction on the depth $h(\Gamma _{1})$ of
the decision tree $\Gamma _{1}$. Let $h(\Gamma _{1})=0$. Then, as the
decision tree $\Gamma _{2}$, we can take the decision tree $\Gamma _{1}$. It
is clear that $h(\Gamma _{2})=2^{h(\Gamma _{1})}-1$. We now assume that the
considered statement is true for any binary information system, any problem
over this system, and any decision tree over the considered problem that
solves\ this problem, uses both attributes and proper hypotheses, and has
depth at most $k$, $k\geq 0$.

Let $U=(A,F)$ be a binary information system, $z=(f_{1},\ldots ,f_{n})$ be a
problem over $U$, and $\Gamma _{1}$ be a decision tree over $z$ that solves\
the problem $z$\ relative to $U$, uses both attributes from $F(z)$ and
proper hypotheses over $z$, and satisfies the condition $h(\Gamma _{1})=k+1$%
. We now show that there exists a decision tree $\Gamma _{2}$ over $z$,
which solves\ the problem $z$\ relative to $U$, uses only proper hypotheses
over $z$, and which depth is at most $2^{k+1}-1$.

Let the root of $\Gamma _{1}$ be labeled with a proper hypothesis $%
H=\{f_{1}(x)=\delta _{1},\ldots ,f_{n}(x)=\delta _{n}\}$. Then there are $%
n+1 $ edges, which leave the root, are labeled with the systems of equations
$H$, $\{f_{1}(x)=\lnot \delta _{1}\}$, ..., $\{f_{n}(x)=\lnot \delta _{n}\}$%
, and enter the roots of subtrees $G_{0},G_{1},\ldots ,G_{n}$ of the tree $%
\Gamma _{1}$, respectively. It is clear that, for $i=1,\ldots ,n$,  $G_{i}$ is a
decision tree over $z$, which  solves the problem $z$ relative to the binary
information system $U_{i}=(A(f_{i},\lnot \delta_i ),F)$, uses only attributes
and proper hypotheses for $U_{i}$, and satisfies the inequality $%
h(G_{i})\leq k$. Using the inductive hypothesis, we obtain that, for $%
i=1,\ldots ,n$, there exists a decision tree $G_{i}^{\prime }$ over $z$
that solves\ the problem $z$\ relative to $U_{i}$, uses only proper
hypotheses for $U_{i}$, and satisfies the inequalities $h(G_{i}^{\prime
})\leq 2^{h(G_{i})}-1\leq 2^{k}-1$. Let $G_{0}^{\prime }$ be the decision tree, which contains only one node labeled with the tuple $(\delta_1 , \ldots , \delta_n)$.
We denote by $\Gamma _{2}$ the decision
tree over $z$ that is obtained from the decision tree $\Gamma _{1}$ by
replacing the subtrees $G_{0},G_{1},\ldots ,G_{n}$ with the subtrees $%
G_{0}^{\prime },G_{1}^{\prime },\ldots ,G_{n}^{\prime }$. It is easy to show that $\Gamma
_{2}$ is a decision tree over $z$, which solves the problem $z$ relative to $U$, uses only proper hypotheses
for $U$, and satisfies the inequalities $h(\Gamma _{2})\leq 2^{k}-1+1\leq
2^{h(\Gamma _{1})}-1$.

Let the root of $\Gamma _{1}$ be labeled with an attribute $f_{i}$. Then
there are two edges, which leave the root, are labeled with
the
systems of equations $\{f_{i}(x)=0\}$ and $\{f_{i}(x)=1\}$, and enter the roots of subtrees $T_{0}$ and $T_{1}$ of the tree $%
\Gamma _{1}$, respectively. It is clear that, for $p=0,1$, $T_{p}$ is a decision tree over $z$, which
 solves the problem $z$ relative to the binary information system $%
U_{p}=(A(f_{i},p),F)$, uses only attributes and proper hypotheses for $U_{p}$%
, and satisfies the inequality $h(T_{p})\leq k$. Using the inductive
hypothesis, we obtain that, for $p=0,1$, there exists a decision tree $%
T_{p}^{\prime }$ over $z$ that solves\ the problem $z$\ relative to $%
U_{p} $, uses only proper hypotheses for $U_{p}$, and satisfies the
inequalities $h(T_{p}^{\prime })\leq 2^{h(T_{p})}-1\leq 2^{k}-1$. We denote
by $T$ the decision tree obtained from the decision tree $T_{0}^{\prime }$
by replacing each terminal node of $T_{0}^{\prime }$ with the decision
tree $T_{1}^{\prime }$.

Denote by $\Gamma _{2}$ the decision tree
obtained from $T$ by the following transformation of each complete path $\xi
$ in $T$. If $\mathcal{A}(\xi )=\emptyset $, then we keep the path $\xi $
untouched. Let $\mathcal{A}(\xi )\neq \emptyset $, $\bar{\delta}=(\delta
_{1},\ldots ,\delta _{n})$ be the tuple that was attached to the terminal
node of the tree $T_{0}^{\prime }$ through which the path $\xi $ passes,
and $\bar{\sigma}=(\sigma _{1},\ldots ,\sigma _{n})$ be the tuple attached
to the terminal node of $\xi $. Since $\mathcal{A}(\xi )\neq\emptyset $, at
least one of the tuples $\bar{\delta}$ and $\bar{\sigma}$ belongs to the set
$\Delta _{U}(z)$. Let, for the definiteness, $\bar{\delta}\in \Delta _{U}(z)$%
. Denote $H=\{f_{1}(x)=\delta _{1},\ldots ,f_{n}(x)=\delta _{n}\}$. We
replace the terminal node of the path $\xi $ with the working node labeled
with the hypothesis $H$, which is proper for $U$. There are $n+1$ edges that
leave this node and are labeled with the systems of equations $%
H,\{f_{1}(x)=\lnot \delta _{1}\},...,\{f_{n}(x)=\lnot \delta _{n}\}$,
respectively. The edge labeled with $H$ enters to the terminal node labeled
with the tuple $\bar{\delta}$. All other edges enter to terminal nodes
labeled with the tuple $\bar{\sigma}$. One can show that $\Gamma _{2}$ is a
decision tree over $z$ that solves\ the problem $z$\ relative to $U$, uses
only proper hypotheses for $U$, and satisfies the relations $h(\Gamma
_{2})\leq 2(2^{k}-1)+1=2^{h(\Gamma _{1})}-1$.
\end{proof}

\begin{lemma}
\label{L3a} Let $U=(A,F)$ be an infinite binary information system. Then $%
h_{U}^{(3)}(n)\leq h_{U}^{(5)}(n)\leq h_{U}^{(4)}(n)\leq n$ and $%
h_{U}^{(2)}(n)\leq h_{U}^{(4)}(n)$ for any $n\in \mathbb{N}$.
\end{lemma}

\begin{proof}
It is clear, that $h_{U}^{(3)}(z)\leq h_{U}^{(5)}(z)\leq h_{U}^{(4)}(z)$ and
$h_{U}^{(2)}(z)\leq h_{U}^{(4)}(z)$ for any problem $z$ over $U$. Therefore $%
h_{U}^{(3)}(n)\leq h_{U}^{(5)}(n)\leq h_{U}^{(4)}(n)$ and $%
h_{U}^{(2)}(n)\leq h_{U}^{(4)}(n)$ for any $n\in \mathbb{N}$.

We now consider an arbitrary problem $z=(f_{1},\ldots ,f_{n})$ over $U$ and
a decision tree over $z$, which uses only proper hypotheses for $U$ and solves the problem $z$ relative to $U$ in the following way. For
a given element $a\in A$, the first query is about an arbitrary proper
hypothesis $H_{1}=\{f_{1}(x)=\delta _{1},\ldots ,f_{n}(x)=\delta _{n}\}$ for
$U$. If the answer is $H_{1}$, then the problem $z$ is solved for the
element $a$. If, for some $i\in \{1,\ldots ,n\}$, the answer is $%
\{f_{i}(x)=\lnot \delta _{i}\}$, then the second query is about a proper
hypothesis $H_{2}=\{f_{1}(x)=\sigma _{1},\ldots ,f_{n}(x)=\sigma _{n}\}$
such that $\sigma _{i}=\lnot \delta _{i}$. If the answer is $H_{2}$, then
the problem $z$ is solved for the element $a$. If, for some $j\in \{1,\ldots
,n\}$, the answer is $\{f_{j}(x)=\lnot \sigma _{j}\}$, then the third query
is about a proper hypothesis $H_{3}=\{f_{1}(x)=\gamma _{1},\ldots
,f_{n}(x)=\gamma _{n}\}$ such that $\gamma _{i}=\lnot \delta _{i}$ and $%
\gamma _{j}=\lnot \sigma _{j}$, etc. It is clear that after at most $n$
queries the problem $z$ for the element $a$ will be solved. Thus, $%
h_{U}^{(4)}(z)\leq n$. Since $z$ is an arbitrary problem over $U$, we have $%
h_{U}^{(4)}(n)\leq n$ for any $n\in \mathbb{N}$.
\end{proof}

\begin{proof}[Proof of Theorem \protect\ref{T4a}]
(a) Let $r\in \mathbb{N}$. We now show by induction on $k\in \mathbb{N}\cup
\{0\}$ that, for each binary $r$-i-reduced $k$-information system $U
$ (not necessary infinite) for each problem $z$ over $U$, the inequality $%
h_{U}^{(5)}(z)\leq rk$ holds.

Let $U=(A,F)$ be a binary $r$-i-reduced $0$%
-information system and $z$ be a problem over $U$. Since all attributes from
$F(z)$ are constant on $A$, the set $\Delta _{U}(z)$ contains only one
tuple. Therefore the decision tree consisting of one node labeled with
this tuple solves the problem $z$ relative to $U$, and $h_{U}^{(5)}(z)=0$.

Let $k\in \mathbb{N}\cup \{0\}$ and, for each $m$, $0\leq m\leq k$, the
considered statement hold. Let us show that it holds for $k+1$. Let $U=(A,F)$
be a binary $r$-i-reduced $(k+1)$-information system and $z=(f_{1},\ldots
,f_{n})$ be a problem over $U$. For $i=1,\ldots ,n$, choose a number $\delta
_{i}\in \{0,1\}$ such that the information system $(A(f_{i},\lnot \delta
_{i}),F)$ is $m_{i}$-information system where $1\leq m_{i}\leq k$. It is
easy to show that $(A(f_{i},\lnot \delta _{i}),F)$ is $r$-i-reduced
information system. Using the inductive hypothesis, we conclude that, for $%
i=1,\ldots ,n$, there is a decision tree $\Gamma _{i}$ over $z$, which uses
both attributes from $F(z)$ and proper hypotheses for $%
(A(f_{i},\lnot \delta _{i}),F)$, solves the problem $z$ relative to $(A(f_{i},\lnot
\delta _{i}),F)$, and has depth at most $rm_{i}$.

Let the hypothesis $H=\{f_{1}(x)=\delta _{1},\ldots ,f_{n}(x)=\delta _{n}\}$
be proper for $U$. We denote by $T_{1}$ a decision tree in which the root is
labeled with the hypothesis $H$, the edge leaving the root and labeled with $%
H$ enters the terminal node labeled with the tuple $(\delta _{1},\ldots
,\delta _{n})$, and for $i=1,\ldots ,n$, the edge leaving the root and
labeled with $\{f_{i}(x)=\lnot \delta _{i}\}$ enters the root of the tree $%
\Gamma _{i}$. One can show that $T_{1}$ is a decision tree over  $z$, which uses both attributes and proper hypotheses for $U$,
solves the problem $z$ relative to $%
U$, and satisfies the inequalities $h(T_{1})\leq rk+1\leq r(k+1)$.

Let the hypothesis $H$ be not proper for $U$. Then the equation system $%
\{f_{1}(x)=\delta _{1},\ldots ,f_{n}(x)=\delta _{n}\}$ is inconsistent on $A$%
, and there exists its subsystem $\{f_{i_{1}}(x)=\delta _{i_{1}},\ldots
,f_{i_{t}}(x)=\delta _{i_{t}}\}$, which is inconsistent on $A$ and for which
$t\leq r$. We denote by $G$ a decision tree over $z$ with $2^{t}$ terminal
nodes in which each terminal node is labeled with $n$-tuple $(0,\ldots ,0)$,
and each complete path contains $t$ working nodes labeled with attributes $%
f_{i_{1}},\ldots ,f_{i_{t}}$ starting from the root. We denote by $T_{2}$ a decision tree obtained from the decision tree $G$
by transformation of each complete path $\xi $ in $G$. Let
$\{f_{i_1}(x)=\sigma _{1}\}, \ldots , \{f_{i_t}(x)=\sigma _{t}\}$
be equation systems
attached to edges leaving the
working nodes of $\xi $ labeled with the attributes $f_{i_{1}},\ldots
,f_{i_{t}}$, respectively. If $(\sigma _{1},\ldots ,\sigma _{t})=(\delta
_{i_{1}},\ldots ,\delta _{i_{t}})$, then we keep the path $\xi $ untouched.
Otherwise, let $j$ be the minimum number from the set $\{1,\ldots ,t\}$ such
that $\sigma _{j}=\lnot \delta _{i_{j}}$. In this case, we replace the
terminal node of the path $\xi $ with the root of the decision tree $\Gamma
_{i_{j}}$. One can show that $T_{2}$ is a decision tree over  $z$, which uses both attributes and proper hypotheses for $U$,
solves the problem $z$ relative to $%
U$, and satisfies the inequalities $h(T_{2})\leq rk+t\leq r(k+1)$. Therefore, $h_{U}^{(5)}(z)\leq r(k+1)$
for any problem $z$ over $U$.

Let $U\in \mathcal{C\cap I}$. Then $U$ is $r$-i-reduced $k$-information
system for some natural $r$ and $k$, and $h_{U}^{(5)}(z)\leq rk$ for each problem $z$ over $U$. From Lemma \ref{L2a} it follows that $h_{U}^{(4)}(z)\leq 2^{rk}-1$ for each
problem $z$ over $U$. Therefore $%
h_{U}^{(4)}(n)=O(1)$ and $h_{U}^{(5)}(n)=O(1)$.

(b) Let $U=(A,F)\in (\mathcal{D}\setminus \mathcal{C)\cap I}$. By Lemma \ref{L3a},
$h_{U}^{(5)}(n)\geq h_{U}^{(3)}(n)$ and $h_{U}^{(4)}(n)\geq h_{U}^{(2)}(n)$
for any $n\in \mathbb{N}$. Using the fact that $U\in \mathcal{D}\setminus
\mathcal{C}$ and Theorem \ref{T2a}, we obtain $h_{U}^{(2)}(n)=\Omega (\log
n) $ and $h_{U}^{(3)}(n)=\Omega (\frac{\log n}{\log \log n})$. Therefore $%
h_{U}^{(4)}(n)=\Omega (\log n)$ and $h_{U}^{(5)}(n)=\Omega (\frac{\log n}{%
\log \log n})$.

Since the information systems $U$ belongs to the set $%
\mathcal{D}$, it has finite $I$-dimension $I(U)$. Since $U\in \mathcal{I}$,
the information system $U$ is $r$-i-reduced for some natural $r$. We assume that $r\geq 2$. We
can do it because each $t$-i-reduced information system, $t\in \mathbb{N}$,
is $(t+1)$-i-reduced.

We now show that $h_{U}^{(4)}(n)=O(\log n)$. Let $%
z=(f_{1},\ldots ,f_{n})$ be an arbitrary problem over $U$. From Lemma 5.1
\cite{Moshkov05} it follows that $|\Delta _{U}(z)|\leq (4n)^{I(U)}$. The
proof of this lemma is based on the results similar to ones obtained by
Sauer \cite{Sauer72} and Shelah \cite{Shelah72}.

We consider a decision tree $\Gamma $ over $z$, which solves the problem $z$ relative to
$U$ and uses only proper hypotheses for $U$. This tree is constructed by a variant
of the halving algorithm  \cite{Angluin04,Hegedus95,Hellerstein96}. We describe the
work of this tree for an arbitrary element $a$ from $A$. Set $\Delta =$ $%
\Delta _{U}(z)$. If $|\Delta |=1$, then the only $n$-tuple from $\Delta $ is
the solution $z(a)$ to the problem $z$ for the element $a$. Let $|\Delta
|\geq 2$. For $i=1,\ldots ,n$, we denote by $\delta _{i}$ a number from $%
\{0,1\}$ such that $|\Delta (f_{i},\delta _{i})|\geq |\Delta (f_{i},\lnot
\delta _{i})|$.

Let the system of equations $H=\{f_{1}(x)=\delta _{1},\ldots
,f_{n}(x)=\delta _{n}\}$ be consistent on $A$. In this case, the root of $%
\Gamma $ is labeled with the proper hypothesis $H$. After this query, either the
problem $z$ will be solved (if the answer is $H$) or the number of remaining
tuples in $\Delta $ will be at most $|\Delta |/2$ (if the answer is a
counterexample $\{f_{i}(x)=\lnot \delta _{i}\}$).

Let the system of equations $H$ be inconsistent on $A$. For any inconsistent
subsystem $B$ of $H$, there exists a subsystem $C$ of $B$, which is
inconsistent and contains at most $r$ equations. Then the system $C$
contains at least one equation $f_{i}(x)=\delta _{i}$ such that $|\Delta
(f_{i},\lnot \delta _{i})|\geq |\Delta |/r$. If we assume
the contrary, we obtain that the system $C$ is consistent, which is
impossible. Let $f_{i}\in \{f_{1},\ldots ,f_{n}\}$. The attribute $f_{i}$ is
called balanced if $|\Delta (f_{i},\lnot \delta _{i})|\geq |\Delta
|/r$, and unbalanced if $|\Delta (f_{i},\lnot \delta
_{i})|<|\Delta |/r$.

We denote by $H_{u}$ the subsystem
of $H$ consisting of all equations $f_{i}(x)=\delta _{i}$ from $H$ with
unbalanced attributes $f_{i}$. We now show that the system $H_{u}$ is consistent. Let us
assume the contrary. Then it will contain at least one equation for balanced
attribute, which is impossible. Let $a$ be a solution from $A$ to the
system $H_{u}$, and $f_{1}(a)=\sigma _{1},\ldots ,f_{n}(a)=\sigma _{n}$.
Then the system of equations $P=\{f_{1}(x)=\sigma _{1},\ldots
,f_{n}(x)=\sigma _{n}\}$ is consistent on $A$.

In the considered case, the
root of $\Gamma $ is labeled with the proper hypothesis $P$. After this query,
either the problem $z$ will be solved (if the answer is $P$), or the number
of remaining tuples in $\Delta $ will be less than $|\Delta |/r$ (if the
answer is a counterexample $\{f_{i}(x)=\lnot \sigma _{i}\}$ and $f_{i}$ is
an unbalanced attribute), or the number of remaining tuples in $\Delta $
will be at most $|\Delta |/2$ (if the answer is a counterexample $%
\{f_{i}(x)=\lnot \sigma _{i}\}$, $\sigma _{i}=\delta _{i}$, and $f_{i}$ is a
balanced attribute), or the number of remaining tuples in $\Delta $ will be
at most $|\Delta |(1-1/r)$ (if the answer is a counterexample $%
\{f_{i}(x)=\lnot \sigma _{i}\}$, $\sigma _{i}=\lnot \delta _{i}$, and $f_{i}$
is a balanced attribute).

After the first query ($H$ or $P$) of the decision tree $\Gamma $%
, either the problem $z$ will be solved or the number of remaining tuples in
$\Delta $ will be at most $|\Delta |(1-1/r)$. In the latter case, when the
answer is a counterexample of the kind $\{f_{i}(x)=\lnot \gamma _{i}\}$ ($%
\gamma _{i}=\delta _{i}$ if the first query is $H$ and $\gamma _{i}=\sigma
_{i}$ if the first query is $P$) set $\Delta =$ $\Delta _{U}(z)(f_{i},\lnot
\gamma _{i})$. It is easy to show that the information system $%
(A(f_{i},\lnot \gamma _{i}),F)$ is also $r$-i-reduced. The decision tree $%
\Gamma $ continues to work with the element $a$ and the set of $n$-tuples $%
\Delta $ in the same way.

Let during the work with the element $a$, the
 decision tree $\Gamma$ make $q$ queries. After the $(q-1)$th query, the
number of remaining $n$-tuples in the set $\Delta $ is at least two and at
most $(4n)^{I(U)}(1-1/r)^{q-1}$. Therefore $(1+1/(r-1))^{q}\leq (4n)^{I(U)}$
and $q\ln (1+1/(r-1))\leq I(U)\ln (4n)$. Taking into account that $\ln
(1+1/m)>1/(m+1)$ for any natural $m$, we obtain $q\,<rI(U)\ln (4n)$. So
during the processing of the element $a$, the decision tree $\Gamma $ makes
at most $rI(U)\ln (4n)$ queries. Since $a$ is an arbitrary element from $A$,
the depth of $\Gamma $ is at most $rI(U)\ln (4n)$ and $h_{U}^{(4)}(z)\le rI(U)\ln (4n)$. Since $z$ is an arbitrary
problem over $U$, we obtain $h_{U}^{(4)}(n)=O(\log n)$. By Lemma \ref{L3a}, $h_{U}^{(5)}(n)=O(\log n)$.

(c) Let $U=(A,F)\in \mathcal{D}\setminus \mathcal{I}$. From Lemma \ref{L3a}
it follows that $h_{U}^{(5)}(n)\leq h_{U}^{(4)}(n)\leq n$ for any $n\in
\mathbb{N}$. We now show that, for any  $m\in \mathbb{N}$, there exists a natural $n$ such that $n \ge m$, $h_{U}^{(4)}(n)\ge n-1$, and $h_{U}^{(5)}(n)\ge n-1$.

Let $m\in \mathbb{N}$. Since $U\notin \mathcal{I}$, there
exists a system of equations $P$ over $U$ with $n\geq m$ equations such that
$P$ is inconsistent but each proper subsystem of $P$ is consistent on $A$.
Let, for the definiteness, $P=\{f_{1}(x)=0,\ldots ,f_{n}(x)=0\}$. Consider
the problem $z=(f_{1},\ldots ,f_{n})$ over $U$. Then, for $i=1,\ldots ,n$,
the set $\Delta _{U}(z)$ contains $n$-tuple $\bar{\delta}_{i}=(0,\ldots
,0,1,0,\ldots ,0)$ in which all digits with the exception of the $i$th one
are equal to $0$.

Let $\Gamma $ be a decision tree over $z$ that solves the problem $z$ relative to $U$ and uses both attributes and proper hypotheses for $U$.
We
consider a complete path $\xi $ in $\Gamma $ in which each edge is labeled with an equation system of the kind $\{f_i(x)=0\}$, where $f_i \in F(z)$. Such complete path exists since $P$ is not a proper hypothesis. Let  $\pi (\xi
)=(f_{i_{1}},0)\cdots (f_{i_{t}},0)$ for some $f_{i_{1}},\ldots ,f_{i_{t}}\in F(z)$.
Since $\Gamma $ solves the problem $z$, the set $\Delta
_{U}(z)\pi (\xi )$ contains at most one tuple. If we assume that $t<n-1$, we
obtain that $\Delta _{U}(z)\pi (\xi )$ contains at least two tuples.
Therefore $t\geq n-1$ and $h(\Gamma )\geq n-1$. Thus, $h_{U}^{(5)}(z)\geq
n-1 $, $h_{U}^{(5)}(n)\geq n-1$ and, by Lemma \ref{L3a},  $h_{U}^{(4)}(n)\geq n-1$.

(d) Let $U\notin \mathcal{D}$. From Lemma \ref{L3a} it follows that $%
h_{U}^{(3)}(n)\leq h_{U}^{(5)}(n)\leq h_{U}^{(4)}(n)\leq n$ for any $n\in
\mathbb{N}$. By Theorem \ref{T2a}, $h_{U}^{(3)}(n)=n$ for any $n\in
\mathbb{N}$. Thus, $h_{U}^{(5)}(n)=h_{U}^{(4)}(n)=n$ for any $n\in \mathbb{N}
$.

\end{proof}

\section{Proof of Theorem \protect\ref{T5a} \label{S6a}}

First, we consider several auxiliary statements.

\begin{lemma}
\label{P1a} $\mathcal{R}\subseteq \mathcal{I}$.
\end{lemma}

\begin{proof}
Let $U=(A,F)\in \mathcal{R}$. Then $U$ is $r$-restricted for some natural $r$%
. We now show that $U$ is $(r+1)$-i-restricted. Let $S$ be an arbitrary
inconsistent on $A$ equation system over $U$ and $S^{\prime }$ be a
subsystem of $S$ with the maximum number of equations that is consistent.
Since $U$ is $r$-restricted, the system $S^{\prime }$ has a subsystem $%
S^{\prime \prime }$ with at most $r$ equations and the same set of solutions
on $A$ as the system $S^{\prime }$. It is clear that there exists an
equation $f(x)=\delta $ from $S$ such that the system of equations $%
S^{\prime }\cup \{f(x)=\delta \}$ is inconsistent. Then the subsystem $%
S^{\prime \prime }\cup \{f(x)=\delta \}$ of $S$ with at most $r+1$ equations
is inconsistent. Therefore $U$ is $(r+1)$-i-restricted and $U\in \mathcal{I}$%
.
\end{proof}

\begin{table}[h]
\caption{All extensions of rows of Table \protect\ref{tab1}}
\label{tab5}\center
\begin{tabular}{|l|llll|}
\hline
& $\mathcal{R}$ & $\mathcal{D}$ & $\mathcal{C}$ & $\mathcal{I}$ \\ \hline
1 & $0$ & $0$ & $0$ & $0$ \\
2 & $0$ & $0$ & $0$ & $1$ \\
3 & $0$ & $1$ & $0$ & $0$ \\
4 & $0$ & $1$ & $0$ & $1$ \\
5 & $0$ & $1$ & $1$ & $0$ \\
6 & $0$ & $1$ & $1$ & $1$ \\
7 & $1$ & $1$ & $0$ & $1$ \\
8 & $1$ & $1$ & $0$ & $0$ \\ \hline
\end{tabular}%
\end{table}

\begin{lemma}
\label{P2a} For any infinite binary information system, its extended
indicator vector coincides with one of the rows of Table \ref{tab3}.
\end{lemma}

\begin{proof}
Let $U$ be an infinite binary information system and $$%
eind(U)=(c_{1},c_{2},c_{3},c_{4}).$$ Then $ind(U)=(c_{1},c_{2},c_{3})$. By
Theorem \ref{T3a}, $(c_{1},c_{2},c_{3})$ is a row of Table \ref{tab1}.
Therefore, for each infinite binary information system, its extended
indicator vector is an extension of a row of Table \ref{tab1}: it can be
obtained from the row by adding the fourth digit, which is equal to $0$ or $1
$. Table \ref{tab5} contains all extensions of rows of Table \ref{tab1}.
We now show that the row with number 8 cannot be the
extended indicator vector of an infinite binary information system. Assume the
contrary: there is an infinite binary information system $U^{\prime }$ such that $%
eind(U^{\prime })=(1,1,0,0)$. Then $U^{\prime }\in \mathcal{R}$ and $U^{\prime }\notin \mathcal{I}$, but
this is impossible since, by Lemma \ref{P1a}, $\mathcal{R}\subseteq \mathcal{%
I}$. Therefore, for any infinite binary information system, its extended
indicator vector coincides with one of the rows of Table \ref{tab5} with
numbers 1-7. Thus, it coincides with one of the rows of Table \ref{tab3}.

\end{proof}

Let $d\in \mathbb{N}$. A $d$-\emph{complete tree over a binary information
system} $U=(A,F)$ is a marked finite directed tree with the root in which

\begin{itemize}
\item Each terminal node is not labeled.

\item Each nonterminal node is labeled with an attribute $f\in F$. There are
two edges leaving this node that are labeled with the systems of equations $%
\{f(x)=0\}$ and $\{f(x)=1\}$, respectively.

\item The length of each complete path (path from the root to a terminal
node) is equal to $d$.

\item For each complete path $\xi $, the equation system $\mathcal{S}(\xi )$%
, which is the union of equation systems assigned to the edges of the path $%
\xi $, is consistent.
\end{itemize}

Let $G$ be a $d$-complete tree over $U$. We denote by $F(G)$ the set of all attributes
attached to the nonterminal nodes of the tree $G$. The following statement was presented in \cite{DAM}.

\begin{lemma}
\label{L0b} {\rm \cite{DAM}} Let $U=(A,F)$ be a binary information system, $d\in
\mathbb{N}$, $G$ be a $d$-complete tree over $U$, and $z$ be a problem over $%
U$ such that $F(G)\subseteq F(z)$. Then $h_{U}^{(2)}(z)\geq d$.
\end{lemma}

We now define seven infinite binary information systems $U_1, \ldots , U_7$ and prove that these systems belong to the complexity classes $\mathcal{V}_{1}, \ldots , \mathcal{V}_{7}$, respectively. Four of these systems were considered in \cite{DAM}. For the completeness, we repeat some reasonings from this paper.

Define an infinite binary information system $U_{1}=(A_{1},F_{1})$ as
follows: $A_{1}=\mathbb{N}$ and $F_{1}$ is the set of all functions from $%
\mathbb{N}$ to $\{0,1\}$.

\begin{lemma}
\label{L1b}The information system $U_{1}$ belongs to the class $\mathcal{V}_{1}$.
\end{lemma}

\begin{proof}
It is easy to show that the information system $U_{1}$ has infinite $I$%
-dimension. Therefore $U_{1}\notin \mathcal{D}$. We now show that $%
U_{1}\notin \mathcal{I}$. For a natural $n$, we define functions $%
f_{0},f_{1},\ldots ,f_{n}\in F_{1}$. For any $a\in \mathbb{N}$, $f_{0}(a)=1$
if and only if $a\in \{1,\ldots ,n\}$. For $i=1,\ldots ,n$, $f_{i}(a)=1$ if
and only if $a=i$. It is easy to show that the equation system $%
\{f_{0}(x)=1,f_{1}(x)=0,\ldots ,f_{n}(x)=0\}$ is inconsistent but each
proper subsystem of this system if consistent. Therefore $U_{1}\notin
\mathcal{I}$. Using Lemma \ref{P2a}, we obtain $eind(U_{1})=(0,0,0,0)$, i.e.,
$U_{1}\in \mathcal{V}_{1}$.
\end{proof}

Define an infinite binary information system $U_{2}=(A_{2},F_{2})$ as
follows: $A_{2}$ is the set of all infinite sequences of the kind $%
a_{1},a_{2},\ldots $, where $a_{j}\in \{0,1\}$, $j\in \mathbb{N}$, $%
F_{2}=\{f_{i}:i\in \mathbb{N}\}$ and $f_{i}(a_{1},a_{2},\ldots )=a_{i}$.

\begin{lemma}
\label{L2b}The information system $U_{2}$ belongs to the class $\mathcal{V}%
_{2}$.
\end{lemma}

\begin{proof}
It is easy to show that the information system $U_{2}$ has infinite $I$%
-dimension. Therefore $U_{2}\notin \mathcal{D}$. Let $S$ be a system of
equations over $U_{2}$. It is clear that $S$ is inconsistent if and only if,
for some $i\in \mathbb{N}$, the system $S$ contains equations $f_{i}(x)=0$
and $f_{i}(x)=1$. Therefore $U_{2}$ is $2$-i-restricted and $U_{2}\in \mathcal{I}$. Using Lemma \ref{P2a},
we obtain $eind(U_{2})=(0,0,0,1)$, i.e., $U_{2}\in \mathcal{V}_{2}$.
\end{proof}

For any $i\in \mathbb{N}$, we define two functions $p_{i}:\mathbb{N}%
\rightarrow \{0,1\}$ and $l_{i}:\mathbb{N}\rightarrow \{0,1\}$. Let $j\in
\mathbb{N}$. Then $p_{i}(j)=1$ if and only if $j=i$, and $l_{i}(j)=1$ if and
only if $j>i$.

Define an infinite binary information system $U_{3}=(A_{3},F_{3})$ as
follows: $A_{3}=\mathbb{N}$ and $F_{3}=\{p_{i}:i\in \mathbb{N}\}\cup
\{l_{i}:i\in \mathbb{N}\}$.

\begin{lemma}
\label{L3b}The information system $U_{3}$ belongs to the class $\mathcal{V}%
_{3}$.
\end{lemma}

\begin{proof}
For $n\in \mathbb{N}$, denote $S_{n}=\{p_{1}(x)=0,\ldots
,p_{n}(x)=0,l_{n}(x)=0\}$. It is easy to show that the equation system $S_{n}$ is
inconsistent and each proper subsystem of $S_{n}$ is consistent. Therefore $%
U_{3}\notin \mathcal{I}$. By Lemma \ref{P1a},  $U_{3}\notin
\mathcal{R}$. Using attributes from the set $\{l_{i}:i\in \mathbb{N}\}$, we
can construct $d$-complete tree over $U_{3}$ for each $d\in \mathbb{N}$. By
Lemma \ref{L0b} and Theorem \ref{T2a}, $U_{3}\notin \mathcal{C}$. One can
show that $I(U_{3})=1$. Therefore $U_{3}\in \mathcal{D}$. Thus, $%
eind(U_{3})=(0,1,0,0)$, i.e., $U_{3}\in \mathcal{V}_{3}$.
\end{proof}

Define an infinite binary information system $U_{4}=(A_{4},F_{4})$ as
follows: $A_{4}=\mathbb{N}^{2}$ and $F_{4}=\{f_{i}:i\in \mathbb{N}\}\cup
\{f_{i,j}:i,j\in \mathbb{N}\}$. For any $a=(p,q)\in A_{4}$ and any $f_{i}\in
F_{4}$, $f_{i}(a)=1$ if and only if $p>i$. For any $a\in A_{4}$ and any $%
f_{i,j}\in F_{4}$, $f_{i,j}(a)=1$ if and only if $a=(i,j)$.

\begin{lemma}
\label{L4b}The information system $U_{4}$ belongs to the class $\mathcal{V}%
_{4}$.
\end{lemma}

\begin{proof}
Let $n\in \mathbb{N}$ and $S_{n}=\{f_{1,1}(x)=0,\ldots ,f_{1,n}(x)=0\}$. It
is easy to show that the system $S_{n}$ is consistent and each proper
subsystem of $S_{n}$ has another set of solutions on $A_{4}$ than the system
$S_{n}$. Therefore $U_{4}\notin \mathcal{R}$.

Using attributes from the set $\{f_{i}:i\in \mathbb{N}\}$, we can construct $%
d$-complete tree over $U_{4}$ for each $d\in \mathbb{N}$. By Lemma \ref{L0b}
and Theorem \ref{T2a}, $U_{4}\notin \mathcal{C}$.

Let $S$ be an equation system over $U_{4}$. One can show that $S$ is
inconsistent if and only if $S$ contains at least one of the following pairs
of equations:

\begin{itemize}
\item $f_{i,j}(x)=0$ and $f_{i,j}(x)=1$;

\item $f_{i,j}(x)=1$ and $f_{k,l}(x)=1$, $(i,j)\neq (k,l)$;

\item $f_{i,j}(x)=1$ and $f_{k}(x)=0$, $i>k$;

\item $f_{i,j}(x)=1$ and $f_{k}(x)=1$, $i\leq k$;

\item $f_{i}(x)=0$ and $f_{j}(x)=1$, $i\leq j$.
\end{itemize}
Therefore $U_{4}$ is $2$-i-restricted and $U_{4}\in \mathcal{I}$. One can
show that $I(U_{4})=1$. Therefore $U_{4}\in \mathcal{D}$. Thus, $%
eind(U_{4})=(0,1,0,1)$, i.e., $U_{4}\in \mathcal{V}_{4}$.
\end{proof}

Define an infinite binary information system $U_{5}=(A_{5},F_{5})$ as
follows: $A_{5}=\bigcup_{i\in \mathbb{N}}\{(i,1),\ldots ,(i,i)\}$ and $%
F_{5}=\bigcup_{i\in \mathbb{N}}\{f_{i},f_{i,1},\ldots ,f_{i,i}\}$. For any $%
a\in A_{5}$ and any $f_{i}\in F_{5}$, $f_{i}(a)=1$ if and only if $a\in
\{(i,1),\ldots ,(i,i)\}$. For any $a\in A_{5}$ and any $f_{i,j}\in F_{5}$, $%
f_{i,j}(a)=1$ if and only if $a=(i,j)$.

\begin{lemma}
\label{L5b}The information system $U_{5}$ belongs to the class $\mathcal{V}%
_{5}$.
\end{lemma}

\begin{proof}
It is easy to show that $U_{5}$ is $2$-information system. In particular, $%
(A_{5}(f_{i},1),F_{5})$ is $0$-information system if $i=1$, $%
(A_{5}(f_{i},1),F_{5})$ is $1$-information system if $i>1$, and $%
(A_{5}(f_{i,j},1),F_{5})$ is $0$-information system for any attribute $%
f_{i,j}\in F_{5}$. Therefore $U_{5}\in \mathcal{C}$. Let $i\in \mathbb{N}$
and $S_{i}=\{f_{i}(x)=1,f_{i,1}(x)=0,\ldots ,f_{i,i}(x)=0\}$. One can show
that $S_{i}$ is inconsistent and each proper subsystem of $S_{i}$ is
consistent. Therefore $U_{5}\notin \mathcal{I}$. Using Lemma \ref{P2a}, we
obtain $eind(U_{3})=(0,1,1,0)$, i.e., $U_{5}\in \mathcal{V}_{5}$.
\end{proof}

Define an infinite binary information system $U_{6}=(A_{6},F_{6})$ as
follows: $A_{6}=\mathbb{N}$ and $F_{6}=\{p_{i}:i\in \mathbb{N}\}$.

\begin{lemma}
\label{L6b}The information system $U_{6}$ belongs to the class $\mathcal{V}%
_{6}$.
\end{lemma}

\begin{proof}
It is easy to show that $U_{6}$ is $1$-information system: evidently, $U_{6}$ is not $0$-information system, and $%
(A_{6}(p_{i},1),F_{6})$ is $0$-information system  for any $i\in \mathbb{N}$. Therefore $%
U_{6}\in \mathcal{C}$. Let $S$ be an equation system over $U_{6}$. One can
show that $S$ is inconsistent if and only if it contains equations $p_{i}(x)=0$ and $%
p_{i}(x)=1$ for some $i\in \mathbb{N}$ or it contains equations $p_{i}(x)=1$
and $p_{j}(x)=1$ for some $i,j\in \mathbb{N}$, $i\neq j$. Therefore $U_{6}$ is $2$-i-restricted and  $%
U_{6}\in \mathcal{I}$. Using Lemma \ref{P2a}, we obtain $%
eind(U_{6})=(0,1,1,1)$, i.e., $U_{6}\in \mathcal{V}_{6}$.
\end{proof}

Define an infinite binary information system $U_{7}=(A_{7},F_{7})$ as
follows: $A_{7}=\mathbb{N}$ and $F_{7}=\{l_{i}:i\in \mathbb{N}\}$.

\begin{lemma}
\label{L7b}The information system $U_{7}$ belongs to the class $\mathcal{V}%
_{7}$.
\end{lemma}

\begin{proof}
Let us consider an arbitrary consistent system of equations $S$ over $U_{7}$%
. We now show that there is a subsystem of $S$, which has at most two
equations and the same set of solutions as $S$. Let $S$ contain both
equations of the kind $l_{i}(x)=1$ and $l_{j}(x)=0$. Denote $i_{0}=\max
\{i:l_{i}(x)=1\in S\}$ and $j_{0}=\min \{j:l_{j}(x)=0\in S\}$. One can show
that the system of equations $S^{\prime }=\{l_{i_{0}}(x)=1,l_{j_{0}}(x)=0\}$
has the same set of solutions as $S$. The case when $S$ contains for some $%
\delta \in \{0,1\}$ only equations of the kind $l_{p}(x)=\delta $ can be
considered in a similar way. In this case, the equation system $S^{\prime }$
contains only one equation. Therefore the information system $U_{7}$ is $2$%
-reduced and $U_{7}\in $ $\mathcal{R}$. Using Lemma \ref{P2a}, we obtain $%
eind(U_{7})=(1,1,0,1)$, i.e., $U_{7}\in \mathcal{V}_{7}$.
\end{proof}

\begin{proof}[Proof of Theorem \protect\ref{T5a}]
From Lemma \ref{P2a} it follows that, for any infinite binary information
system, its extended indicator vector coincides with one of the rows of
Table \ref{tab3}. Using Lemmas \ref{L1b}-\ref{L7b}, we conclude that each
row of Table \ref{tab3} is the extended indicator vector of some infinite
binary information system.  
\end{proof}

\section{Conclusions \label{S7a}}

Based on the results of exact learning and test theory, for an arbitrary
infinite binary information system, we studied five functions, which
characterize the dependence in the worst case of the minimum depth of a
decision tree solving a problem on the number of attributes in the problem
description. These five functions correspond to (i) decision trees using
attributes, (ii) decision trees using arbitrary hypotheses, (iii) decision
trees using both attributes and arbitrary hypotheses, (iv) decision trees
using proper hypotheses, and (v) decision trees using both attributes and
proper hypotheses. The first three functions were considered in \cite{DAM}.
The last two functions were investigated in this paper: we described
possible types of behavior for each of these two functions. We also studied
joint behavior of the considered five functions and distinguished seven
complexity classes of infinite binary information systems. In
the future, we are planing to translate the obtained results into the
language of exact learning.

\subsection*{Acknowledgments}
Research reported in this publication was supported by King Abdullah
University of Science and Technology (KAUST).

\bibliographystyle{spmpsci}
\bibliography{hyp}

\end{document}